\documentclass[a4paper,UKenglish,cleveref, autoref, thm-restate]{lipics-v2021}

\usepackage{xspace}
\usepackage{graphicx} 
\usepackage{amsmath}
\usepackage[capitalise]{cleveref}
\usepackage{amsthm}
\usepackage{hyperref}
\usepackage{xcolor}
\usepackage{algorithm}
\usepackage{algpseudocode}
\usepackage{float}
\usepackage{subcaption}
\usepackage{algorithm}
\usepackage{thmtools}
\usepackage[para]{footmisc}
\usepackage{bm}
\usepackage{amsbsy}

\usepackage{framed}

\newcommand\algorithmicupon{\textbf{upon}}
\algblockdefx{Upon}{EndUpon}[1]{\algorithmicupon\ #1\ \textbf{do}}{}
\algnotext{EndUpon}

\crefname{ALG@line}{line}{lines}
\crefalias{ALG@line}{line}

\makeatletter
\AddToHook{env/algorithmic/before}{\def\@currentcounter{ALG@line}}
\makeatother
\nolinenumbers

\title{Amortized Asynchronous Byzantine Reliable Broadcast with Optimal Resilience} 

\author{Michael Yiqing Hu}{National University of Singapore,}{}{}{}

\author{Alvin Yang Hong Yao}{National University of Singapore, }{}{}{}

\author{Li Jialin}{National University of Singapore,}{}{}{}



\begin{document}

\maketitle

\begin{abstract}
Byzantine Reliable Broadcast (BRB) is a fundamental primitive in distributed computing and cryptographic systems; reducing the communication cost of BRB thus remains an important research direction. 
However, most existing works either focus strictly on the synchronous network model or utilize computationally impractical erasure codes.
Therefore, to achieve a practical yet network-robust algorithm, one must turn toward committee sampling techniques. However, Committee sampling techniques often forgo optimal resilience ($f < \lfloor\frac{n}{3} \rfloor$) in the face of asynchrony. 

This work produces two interesting results: Firstly, we propose a \textit{randomly asynchronous} BRB protocol that can achieve both optimal resilience and asymptotically optimal communication complexity ($O(n|m|)$) through an underutilized technique: \textit{amortization}; and does not utilize computationally expensive \textit{erasure codes}.
Next, we show that an optimally resilient BRB protocol utilizing sampled committees cannot exist in a \textit{fully asynchronous} network.

\end{abstract}

\section{Introduction}
\label{sec:intro}

Byzantine Reliable Broadcast (BRB) serves as the cornerstone of many algorithms in distributed systems, ranging from Verifiable Secret Sharing (VSS)~\cite{Chor1985VerifiableSS,Canetti1993FastAB,Cachin2000RandomOI} and Multi-Party Computation (MPC)~\cite{Katz2006OnEC,Cramer2000GeneralSM}, to State Machine Replication (SMR)~\cite{Keidar2021AllYN,Danezis2021NarwhalAT,Giridharan2022BullsharkDB}.
The BRB primitive ensures that any of the $n$ nodes may disseminate a message $m$ to the remainder of the nodes in the presence of $f$ corrupt/Byzantine nodes, where Byzantine nodes may deviate arbitrarily from the protocol. BRB ensures that all honest (non-Byzantine) parties will agree on the sender's message.

Bracha's BRB protocol~\cite{Bracha1987AsynchronousBA} was the first algorithm that works in full asynchrony and tolerates $f<n/3$ Byzantine faults.
Bracha's solution is information-theoretically secure (assuming signatures) with a communication complexity of $O(n^2|m|)$ where $|m|$ is the size of the message in bits.
Dolev and Reischuk~\cite{Dolev1985BoundsOI} also demonstrated that any deterministic BRB protocol requires at least $\Omega(n^2)$ messages in the worst case.

Various works~\cite{Momose2021OptimalCC,Dolev1983AuthenticatedAF,Wan2023OnTA} explore greater fault tolerance or improved message complexity in the \textit{synchronous} setting.
While there also exist some recent improvements to the communication complexity in the asynchronous setting, they mainly capitalize on either erasure codes or committee sampling.

\subsection{Erasure Codes}
\label{sec:intro:ec}

AVID~\cite{Cachin2005AsynchronousVI} demonstrated how erasure codes might improve communication complexity to $O(n|m| + n^2k \log n)$ when utilized along with a collision-resistant hash function with an output size of $O(k)$.
These schemes typically decompose the original message into smaller fragments to reduce the total message cost to terminate.
Nayak et al.~\cite{Nayak2020ImprovedEP} further improved the communication complexity to $O(n|m|+n^2k)$, terminating in $7$ rounds.
Those bounds were then improved by Alhaddad et al.~\cite{Alhaddad2022BalancedBR}, addressing the uneven distribution of load borne by Das et al. \cite{Das2021AsynchronousDD} to $O(n|m| + nk + n^2 \log n)$.
Notably, AVID~\cite{Cachin2005AsynchronousVI} and Alhaddad et al.~\cite{Alhaddad2022BalancedBR} both claim to be asymptotically optimal if the message is sufficiently large: $|m| = \Omega (n \log n  k)$.

Although BRB optimizations utilizing coding schemes improve the worst-case communication complexity, the cost of all nodes fetching fragments in an attempt to reconstruct the message typically negates any cost savings~\cite{Locher2024ByzantineRB}.
Furthermore, overheads from encoding/decoding and verification degrade performance in real-world systems.
As a result, many practical implementations~\cite{Danezis2021NarwhalAT,Das2021AsynchronousDD,Giridharan2022BullsharkDB} \textit{avoid} coding schemes to optimize BRB communication.

\subsection{Committee Sampling}
\label{intro: comm sampling}

If the computational cost of erasure codes is prohibitive, 
an alternative line of work~\cite{Micali2016ALGORANDTE, Cohen2020NotAC, anikina2024dynamicprobabilisticreliablebroadcast, Shrestha2025TowardsIT, Wang2024PandoES} takes the approach of uniformly sampling a committee $N_c$ of smaller size $n_c \ll n$ and performing reliable broadcast within this subset, after which the result is disseminated to the remaining nodes. 
Since a sampling technique is employed, these algorithms typically succeed with probability at least $1-\varepsilon$  (where $\varepsilon$ is a negligible constant) rather than being purely deterministic.

However, these purely sampling based approaches~\cite{Wang2024PandoES, Cohen2020NotAC} face a fundamental trade-off between scalability and resilience. Because a small random sample is subject to statistical variance, maintaining a high probability of an honest majority within $N_c$ requires a larger margin of honest nodes (sub-optimal resilience) in the global population.
Consequently, protocols \cite{Guerraoui2019ScalableBR} that achieve sub-quadratic communication costs via sampling typically forfeit optimal Byzantine resilience $(n \geq 3f+1)$, while other sample-based approaches~\cite{anikina2024dynamicprobabilisticreliablebroadcast} revert to stronger forms of network synchrony.

Recently, Shrestha and Kate~\cite{Shrestha2025TowardsIT} suggested a hybrid adaptation of Bracha's two-phase BRB~\cite{Bracha1987AsynchronousBA}.
Their protocol first samples $N_c$ nodes while guaranteeing, with all but negligible probability $\varepsilon$, that $N_c$ contains at least $\frac{n_c}{2}$ honest nodes.
The sender communicates $m$ to all nodes in $N_c$, where a simple majority of confirmations ensures that some honest node in the system holds the message.
Subsequently, an all-to-all communication phase involving all $n$ nodes is triggered to establish global agreement that the first phase has completed.

Although this hybrid approach maintains optimal resilience and reduces the cost of the first phase, it does not address the quadratic bottleneck of the second phase, yielding a worst-case total communication cost of $O(n_c^2|m| + n^2k + n_cn|m|)$.



\subsection{Amortized Costs}
\label{sec:intro:amortize}

As illustrated by Wan et al.~\cite{Wan2023OnTA}, amortization remains a largely overlooked technique for improving multi-shot BRB, despite its potential to significantly reduce average system costs over sufficiently long protocol executions.
Their work demonstrates that amortized $O(nk)$ communication complexity in $O(n)$ rounds are achievable in the presence of an honest majority and network synchrony.

Building on this foundation, Civit et al.~\cite{Civit2025RepeatedAI} achieve similar bounds by leveraging a blacklisting approach that progressively excludes Byzantine behavior from the system.
Notably, both techniques~\cite{Wan2023OnTA,Civit2025RepeatedAI} require strong network synchrony to operate.
As further noted in \cite{Civit2025RepeatedAI,Civit2023EveryBC}, amortization analysis for multi-shot tasks remains notably scarce in the literature, particularly in asynchronous settings.

\subsection{The Random Asynchronous Network Model}
\label{sec:intro:randomly async}

A recent work by Danezis et al.~\cite{randomasynch} suggested a novel relaxation of the classic asynchronous network model: the \textit{Random Asynchronous model}. In this network model, messages still preserve an unknown, unbounded delay. However, there does not exist an adversary that may dictate message order; instead, message delivery follows a random schedule. 

Notably, this network model is provably~\cite{randomasynch} \textit{slightly stronger} than full asynchrony, but \textit{slightly weaker} than synchrony. It is also worth mentioning that the random asynchronous model is incomparable to \textit{partial synchrony}. Danezis et al.~\cite{randomasynch} argue that this model is more inline with modern practical systems that typically assume non-adversarial scheduling.

We show in \cref{sec: Impossibility} that an optimally resilient, probabilistic BRB that utilizes smaller $(o(n))$ committees, to exclusively gather acknowledgments to assert that a BRB instance has safely completed, cannot exist in network asynchrony. Therefore, our algorithm utilizes the slightly relaxed \textit{random asynchronous model}.


\begin{figure}[t]
\centering
\scalebox{0.8}{
\begin{tabular}{ccccc}
\hline
Protocol & Network Model & Fault Tolerance & Communication Cost & Technique\\
\hline

Bracha~\cite{Bracha1987AsynchronousBA} & Asynchronous & $n=3f+1$ & $O(n^2|m|)$ & - \\

Chachin-Tessaro~\cite{Cachin2005AsynchronousVI} & Asynchronous & 
$n= 3f+1$ & $O(n|m|+n^2k \log n)$ & Erasure Codes \\

Alhaddad et al.\cite{Alhaddad2022BalancedBR} & Asynchronous & $n=3f+1$& $O(n|m|+nk+n^2\log n)$ & Erasure Codes \\

Wan et al.\cite{Wan2023OnTA} & Synchronous & $n =2f+1$ & $O(nk|m|+n^3k)$ & Amortization \\

Civit et al.\cite{Civit2025RepeatedAI} & Synchronous & $n =2f+1$ & $O(n|m|+nk)$ & Amoritzation \\

Anikina et al.\cite{anikina2024dynamicprobabilisticreliablebroadcast} & Synchronous & $n=3f+1$ & $O(n\log n |m|)$ & Sampling \\

Guerraoui et al.\cite{Guerraoui2019ScalableBR} & Asynchronous & $n\approx6f+1$
 & $O(n\log n |m|)$ & Sampling \\

\hline

\textbf{This work }& Random Asynch. & $n=3f+1$ & $O(n|m|+n^2k)$ & Amortization, Sampling \\

\hline
\end{tabular}
}
\caption{Fault tolerance and communication performance of various BRB protocols, with $|m|$ message size and $k$ bit cryptographic objects. }
\end{figure}

\subsection{Our Contribution}
\label{sec:intro:contributions}

This work explores the following questions:
\begin{enumerate}
    \item \textit{Is there a communication-efficient BRB that operates in the Random Asynchronous network model without utilizing erasure codes?} 
    \item \textit{How far can we push sampling techniques to be utilized in BRB?}
\end{enumerate}

In this paper, we present an amortized probabilistic multi-shot Byzantine Reliable Broadcast algorithm with optimal resilience for the random asynchronous network setting. Our algorithm melds sampling and amortization techniques together, allowing us to yield an amortized cost of $O(n|m|+n^2k)$, where if message size exceeds metadata size ($|m| = \Omega( nk)$), it results in an \textit{asymptotically optimal} communication complexity of $O(n|m|)$.

To the best of our knowledge, our algorithm is the \textit{first} random asynchronous BRB algorithm utilizing sampling that achieves \textit{asymptotically optimal} message complexity while having \textit{optimal resilience} to Byzantine faults. 
This puts it on par with prior works that utilize erasure codes \cite{Alhaddad2022BalancedBR,Das2021AsynchronousDD}.

Our algorithm has a worst-case round cost of $O(n)$ similar to Wan et al.~\cite{Wan2023OnTA}. However, we also provide an optimistic delivery condition that terminates in a single round $(O(1))$ whilst maintaining the same amortized cost. 

\paragraph{Overview of the Paper}

We will begin by stating the system model and preliminaries in \cref{sec: prelims}, followed by an examination of Bracha's~\cite{Bracha1987AsynchronousBA} classic solution in \cref{sec: Case study Brachas} before illustrating our solution in \cref{sec: our algo}. We then prove its correctness in \cref{sec: proofs} and show our algorithm's communication complexity in \cref{sec: complexity analysis}. 
In \cref{sec: Impossibility}, we prove an impossibility result: an optimally resilient, probabilistic BRB that utilizes small $(o(n))$ committees, to exclusively gather acknowledgments to assert that a BRB instance has safely completed, cannot exist in network asynchrony.
We include a discussion on the practical applications of our work in \cref{sec: discussion}, before finally concluding in \cref{sec: conclusion}.

\section{Preliminaries}
\label{sec: prelims}

\subsection{Model and Assumptions}
\label{sec:prelims:model}

We consider a set of $n$ nodes: $\Pi = \{P_1,\cdots,P_n\}$, each attempting to reliably \textit{broadcast a stream} of messages.
A \textit{static} adversary may corrupt $f<\frac{n}{3}$ nodes prior to protocol execution.
A static adversary is conventional and necessary~\cite{abraham19sampling} for BRB algorithms that utilize sampling with unbounded message latency. 
Anikina et al.\cite{anikina2024dynamicprobabilisticreliablebroadcast} enhanced their fault model to one that tolerates a slowly adaptive adversary but relies heavily on network synchrony. 
A non-corrupted party abides by the protocol and is considered to be \textit{honest}; they are considered \textit{Byzantine} otherwise. 
Nodes communicate with each other through message passing via a \textit{random asynchronous} network~\cite{randomasynch}: where messages have unbounded delay but must be delivered eventually. Furthermore, adversarial message scheduling is not possible, and message delivery follows a random schedule; we show in \cref{sec: Impossibility} why this is necessary.

We also make the following assumptions:
\begin{itemize}
    \item A public-key infrastructure (PKI) that supports a digital signature scheme, where every signature can be authenticated.

    \item A collision resistant hash function $\mathcal{H}(x)$ that takes an input $x$ and generates a fixed sized digest of length $k$.

    \item A $(t,N')$ threshold signature scheme~\cite{Boneh2018CompactMF,Boldyreva2002EfficientTS} set up via a trusted dealer or Distributed Key Generation (DKG)~\cite{Boneh2018CompactMF}.
    Each node in $N'$ may generate a signature share $\langle m\rangle_i$ on a message.
    A set of $t$ distinct signature shares on the same message $m$ can be merged to form $cert(m): |cert(m)| = |\langle m\rangle_*| = k$.

    \item The adversary is computationally-bounded and is unable to break the digital signature scheme or the hash function.
\end{itemize}

\subsection{Problem Definition}

This work focuses on the problem of \textit{Probabilistic Multi-shot Byzantine Reliable Broadcast}, which consists of a sequence of single-shot Byzantine reliable broadcasts~\cite{Bracha1987AsynchronousBA} (each with a success probability at least $1-\varepsilon$) from a given node.

\begin{definition}[Byzantine Reliable Broadcast (BRB)]
    A Byzantine Reliable Broadcast with a designated sender and input value $m$, where all honest nodes call $\textsc{Deliver}(m)$, must satisfy the following properties:

    \begin{itemize}
        \item \textbf{Agreement:} If two distinct honest nodes call $\textsc{Deliver}(m')$ and $\textsc{Deliver}(m'')$, respectively, then $m' = m''$.

        \item \textbf{Validity:} If the designated sender is honest, then all honest nodes eventually call $\textsc{Deliver}(m)$.

        \item \textbf{Totality:} If an honest node calls $\textsc{Deliver}(m)$, all honest nodes eventually call $\textsc{Deliver}(m)$.
    \end{itemize}
\end{definition}

We now define our central goal, achieving a Probabilistic Multi-shot Byzantine Reliable Broadcast algorithm:

\begin{definition}[Probabilistic Multi-shot Byzantine Reliable Broadcast ]
\label{def: PMBRB}
    A Probabilistic Multi-shot Byzantine Reliable Broadcast with a designated sender $P_i$ with input value $m$ for some round $r \in \mathbb{N}$, where all honest nodes call $\textsc{Deliver}_i(m,r)$, satisfies the following properties with probability at least $1-\varepsilon$ (where $\varepsilon$ is a negligible constant):

    \begin{itemize}
        \item \textbf{Agreement:} If two distinct honest nodes call $\textsc{Deliver}_i(m',r)$ and $\textsc{Deliver}_i(m'',r)$, respectively, then $m' = m''$.

        \item \textbf{Validity:} If the designated sender is honest, then all honest nodes eventually call \\ $\textsc{Deliver}_i(m,r)$.

        \item \textbf{Totality:} If an honest node calls $\textsc{Deliver}_i(m,r)$, all honest nodes eventually call \\ $\textsc{Deliver}_i(m,r)$.


        \item \textbf{Completion Sequentiality:} If an honest node calls $\textsc{Deliver}_i(\cdot,r)$ where $r>1$, then it must be able to call $\textsc{Deliver}_i(\cdot,r-1)$ eventually.
    \end{itemize}

\end{definition}

Note that our definition deviates slightly from the multi-shot BRB definition of Wan et al.~\cite{Wan2023OnTA}. Specifically, we define our broadcast protocol around a designated sender, whereas \cite{Wan2023OnTA} declares a series of ``slots'' (analogous to our \textit{rounds}) that may be filled by possibly distinct senders.

However, to maintain the sequential and causal broadcast invocations required by many cryptographic protocols that assume a \textit{broadcast channel}~\cite{Gennaro1999SecureDK, Pedersen1991ATC, BenOr1988CompletenessTF}, we retain a slightly modified notion of \textit{Sequentiality}~\cite{Wan2023OnTA}.
In the original formulation, sequentiality ensures that a particular instance of BRB for a round must terminate before BRB for the next round \textit{may begin}.
We modify this property to instead ensure that broadcasts by a given node must terminate in strictly sequential rounds---that is, rounds $1, 2, 3, \dots$ in ascending order, with no reordering permitted. 

\paragraph{Amortized Complexity}
Communication complexity is measured as the number of bits sent by honest nodes.
In this work, we adapt the definition of amortized complexity from Wan et al.~\cite{Wan2023OnTA}.


\begin{definition}[Amortized Communication Complexity]
\label{def: amortized}
    Let $C(r,n,f)$ be the communication complexity of a Probabilistic Multi-shot Byzantine Reliable Broadcast protocol for a designated sender $p_i \in \Pi,|\Pi|=n$, where some honest node calls $\textsc{Deliver}_i(\cdot,r)$ in the presence of $f$ Byzantine nodes.
    The amortized communication complexity of the protocol is given by
\[
\frac{C(r,n,f)}{r}
\]
      

\end{definition}



\section{Case Study of Bracha's BRB}
\label{sec: Case study Brachas}

Bracha's BRB~\cite{Bracha1987AsynchronousBA} operates in two distinct phases.
Notably, most BRB algorithms~\cite{Alhaddad2022BalancedBR,Wang2024PandoES,Shrestha2025TowardsIT,Nayak2020ImprovedEP,Cachin2005AsynchronousVI,Das2021AsynchronousDD,Guerraoui2019ScalableBR} utilize a two-phased structure.
Each phase serves a distinct purpose that cannot be naively eliminated.

In the first phase, the sender broadcasts a message $m$ to all nodes, and each recipient responds by broadcasting an \texttt{echo} message. Once a node receives $n-f$ echo messages, it can guarantee that at least a simple majority of honest nodes have seen the message. However, a single phase is insufficient: due to network asynchrony, only some honest nodes may receive $n-f$ \texttt{echo} messages, leading to inconsistent delivery decisions. Crucially, the first phase prevents equivocation; the quorum intersection property ensures that two distinct messages cannot both complete this phase.

This inadequacy necessitates a second, near-identical phase. Once a node receives $n-f$ \texttt{echo} messages or $f+1$ \texttt{vote} messages, it broadcasts a \texttt{vote} message. Upon receiving $n-f$ vote messages, an honest node will deliver the originally broadcast message $m$.
This ensures that $f+1$ honest nodes have already broadcast a \texttt{vote} message, enabling all honest nodes to eventually deliver, ensuring totality.

\subsection{BRB with a single phase?}

Recent works~\cite{anikina2024dynamicprobabilisticreliablebroadcast, Shrestha2025TowardsIT,Wang2024PandoES,Cohen2020NotAC} on sampling in BRB have improved communication complexity by targeting the first (\texttt{echo}) phase. As described in \cref{intro: comm sampling}, a smaller subset $n_c$ is sampled to provide a probabilistic guarantee that at least one honest node in the system will see the message, even under asynchrony. The responsibility for non-equivocation is then deferred to the subsequent phase.

However, the second (\texttt{vote}) phase, which typically requires a $2f+1$ super-majority is often left untouched. Improving the communication complexity of this second phase is also non-trivial. Committee sampling techniques cannot be applied naively, as a committee with an honest supermajority converges to a size of $n$ when operating with optimal resilience.
Consequently, works~\cite{Wang2024PandoES, Cohen2020NotAC} that do improve this phase via committee sampling typically operate with sub-optimal resilience. For example, Cohen et al.~\cite{Cohen2020NotAC} require $n \approx 6f+1$.

\section{Our approach: Amortized Probabilistic Multi-shot Byzantine Reliable Broadcast}
\label{sec: our algo}

Our approach does not aim to improve the second phase. Instead, we aim to remove it entirely by utilizing a novel sampling-amortization technique. 

Our key insight is as follows: 
\begin{enumerate}
    \item Suppose we sample a committee $N_c$ of size $n_c$ from $\Pi$, such that with at least $1-\varepsilon$ probability, $N_c$ contains a simple majority of honest nodes. 
    \item Getting an acknowledgment response for a message $m$ from a simple $\frac{n_c}{2}+1$ quorum ensures at least a single honest node has witnessed the message. This also ensures that the message is available to all other nodes eventually. 
    \item Repeating item 1 with another unique committee $N'_c$ on the same message $m$ under the random asynchronous model ensures there is a $\frac{2f}{2f+1}$ probability that at least 2 distinct honest nodes have witnessed and asserted that there exists a single unique message $m$.
    \item By repeating this process sufficiently many ($\varphi$) times, we can guarantee with probability at least $1-\varepsilon$ that $f+1$ \textit{distinct} honest nodes have made the same assertion, ensuring non-equivocation on the message.
    \item Furthermore, each new subsequent committee sampling can be piggy-backed to announce a brand new message. This allows amortizing the cost of each BRB instance to a single round of item 1. 
\end{enumerate}

\subsection{Description of APM-BRB}

\begin{algorithm}[t!]
\footnotesize
\caption{Data structures and basic utilities for $P_i$}
\label{algo:our_algo}
\begin{algorithmic}[1]
\State \textbf{Local variables:}
\State \quad \text{triggers}[]: array of $n$ pairs of certificates and round numbers, initialized with \texttt{null}. \label{algo:our_algo:trigger}
\State \quad \text{promises}[][]: array of $n$ lists where promises[$j$] is a sequential list of consecutive messages ordered by round from $P_j$.\label{algo:our_algo:promises}
\State \quad \text{seed}$_0$: value provided by a trusted dealer at protocol start.
\Statex
\Function{sampleCommittee}{$\text{node\_id}, r$}
    \State \textbf{Input:} Node identifier $\text{node\_id}$, round $r$ \quad \textbf{Output:} Subset $N_c \subseteq \Pi$ of size $n_c$
    \State $\text{seed} \leftarrow \mathcal{H}(\text{seed}_0 \| r \bmod \varphi \| \text{node\_id})$; $N_c \leftarrow \emptyset$ 
    \Comment{Committee repeats every $\varphi$ rounds}
    \label{line: our_algo: repeat comm}
    \For{$k = 1$ \textbf{to} $n_c$}
        \State $\text{idx} \leftarrow (\text{seed} \bmod n) + 1$ \Comment{Map seed to node index in $[1, n]$}
        \State $N_c \leftarrow N_c \cup \{P_{\text{idx}}\}$; $\text{seed} \leftarrow \mathcal{H}(\text{seed})$ \Comment{Hash seed for next selection}
    \EndFor
    \State \Return $N_c$
\EndFunction
\Statex
\Function{fetchMsg}{$r, i, \text{useCommittee}$}
    \State $t \leftarrow |\text{promises}[i]|$ \Comment{Current number of messages from $P_i$}
    \For{$y = t+1$ \textbf{to} $r$}
        \State $N_c \leftarrow \texttt{useCommittee} \,?\, \textsc{sampleCommittee}(i, y) : \{P_i\}$
        \If{$y = 1$}
            \State \textbf{request} $m_y^i$ from $N_c$
        \Else
            \If{$y = r$ \textbf{and} \textbf{not} \texttt{useCommittee}} \Comment{Last message, also ask for valid certificate}
                \State \textbf{request} $m_y^i$ from $N_c$ with valid $\text{cert}(m_{y-1}^i) \in m_y^i$ and valid $\text{cert}(m_y^i)$
            \Else
                \State \textbf{request} $m_y^i$ from $N_c$ with valid $\text{cert}(m_{y-1}^i) \in m_y^i$
            \EndIf
        \EndIf
        \State \textbf{wait} until $m_y^i$ received
        \If{$y = r$ \textbf{and} \textbf{not} \texttt{useCommittee}}
            \State \textbf{wait} until $\text{cert}(m_y^i)$ received
        \EndIf
        \State $\text{promises}[i][y] \leftarrow m_y^i$ \label{line: promises_1}
    \EndFor
    \State \Return $\text{promises}[i][r]$
\EndFunction
\Statex
\Function{syncMsg}{$r, i$} \quad \Return \textsc{fetchMsg}$(r, i, \texttt{true})$ \EndFunction
\Statex
\Function{reqMsg}{$r, i$} \quad \Return \textsc{fetchMsg}$(r, i, \texttt{false})$ \EndFunction

\end{algorithmic}
\end{algorithm}
\begin{algorithm}[t]
\footnotesize
\caption{APM-BRB (Sending Routine)}
\label{algo:our_algo_send}
\begin{algorithmic}[1]
\State \textbf{For a sender $P_i$ with input $m_r$ for round $r$:} \Comment{Message sent by each node every round}
\If{$r = 1$}
    \State $N_c \leftarrow \textsc{sampleCommittee}(i, 1)$
    \State $m_1^i \leftarrow \{\text{\textit{message content}}, i, \cdot, \cdot, \text{triggers}[]\}$
    \State \textbf{send} $(m_1^i, r, i)$ to all parties in $N_c$
\Else
    \State $N_c \leftarrow \textsc{sampleCommittee}(i, r)$
    \State \textbf{invoke} $\textsc{reqMsg}(r-1, j)$ for all $j \in [n]$ concurrently; \textbf{wait for} $2f+1$ to complete
    \State $\mathcal{J} \leftarrow$ the set of $2f+1$ nodes, $P_j$ whose $\textsc{reqMsg}(r-1, j)$ completed first
    \State $\text{prev}\{\} \leftarrow \{\text{cert}(\text{promises}[j][r-1]) : j \in \mathcal{J}\}$ \Comment{$2f+1$ certs from distinct nodes in promises[]}
    \State $m_r^i \leftarrow \{\text{message}, i, \text{cert}(m_{r-1}^i), \text{prev}\{\}, \text{triggers}[]\}$ \label{algo:our_algo_send: message params}
    \State \textbf{send} $(m_r^i, r, i)$ to all parties in $N_c$
\EndIf

\State
\State \textbf{For a node $P_j$ (each node in $n$):}
\Upon{receiving $(m_r^i, r, i)$}
    \State $N_c^* \leftarrow \textsc{sampleCommittee}(i, r)$
    \If{$P_j \notin N_c^*$ \textbf{or} already received $(m'^i_r, r)$ s.t. $m'^i_r \neq m_r^i$}
        \State \textbf{ignore}
    \Else
        \If{$\text{cert}(m_{r-1}^i)$ is the $(\frac{n_c}{2}+1, N_c)$ threshold signature of $\textsc{syncMsg}(r-1, i)$} \label{algo:our_algo_send: check extend}
            \State \textbf{send} $(m_r^i, r, i)$ to $N_c^*$ \label{line:resend}
            \State \textbf{send} $(\langle m^i_r \rangle_j, r, i, j)$ to all nodes in $n$ \label{algo:our_algo_send:signature share}
        \Else
            \State \textbf{ignore} \Comment{Ignore if equivocating with local memory}
        \EndIf
    \EndIf
\EndUpon
\end{algorithmic}
\end{algorithm}
\begin{algorithm}[t]
\footnotesize
\caption{APM-BRB (Deliver Condition 1)}
\label{algo:deliver_part1}
\begin{algorithmic}[1]
\State \textbf{For each node in $\Pi$:}
\Upon{receiving $(\langle m^i_r \rangle_j, r, i, j)$ from $\frac{n_c}{2}+1$ distinct nodes in $N_c \leftarrow \textsc{sampleCommittee}(i, r)$ \textbf{or} \text{cert}$(m_r^i)$ from $P_i$} 
    \State $\text{cert}(m_r^i) \leftarrow$ is a valid $(\frac{n_c}{2}+1, N_c)$ threshold signature for $m_r^i$
    \State \textsc{processCertificate}($i, r, \text{cert}(m_r^i)$) \label{algo:deliver_part1:process 1}
\EndUpon
\Upon{\textsc{Deliver}$_i(m, r)$} 
    \For{each $(\text{cert}(m_{r'}^j), r') \in m.\text{triggers}[j]$ that is not \texttt{null}} \label{line:trigger_2}
        \State \textsc{processCertificate}($j, r', \text{cert}(m_{r'}^j)$) \label{algo:deliver_part1:process 2}
    \EndFor
\EndUpon
\State
\Procedure{processCertificate}{$i, r, cert$}
    \If{$r \geq \varphi$}
        \State $v \leftarrow r$; $tempCert \leftarrow cert$
        \If{$\text{triggers}[i] = \texttt{null}$ \textbf{or} round in $\text{triggers}[i] < r$}
            \State $\text{triggers}[i] \leftarrow (cert, r)$ \label{line: trigger_1}
        \EndIf
        \While{$v \geq r - \varphi$} \Comment{Work backwards to get the certificate that should be delivered} \label{line:deliver 1_1}
            \State $N_c \leftarrow \textsc{sampleCommittee}(i, v)$
            \State obtain $m_v^i$ from $N_c$ s.t. $tempCert$ is a $(\frac{n_c}{2}+1, N_c)$ threshold signature for $m_v^i$
            \State $tempCert \leftarrow \text{cert}(m_{v-1}^i) \in m_v^i$; $v \leftarrow v - 1$
        \EndWhile  \label{line:deliver 1_2}
        \For{$u = r-\varphi$ \textbf{down to} $1$} \Comment{Deliver everything that came before the certificate} \label{line: totality_1}
            \State $N_c \leftarrow \textsc{sampleCommittee}(i, u)$
            \State obtain $m_u^i$ from $N_c$ s.t. $tempCert$ is a $(\frac{n_c}{2}+1, N_c)$ threshold signature for $m_u^i$
            \State \textsc{Deliver}$_i(m_u^i, u)$ if not already delivered
            \State $tempCert \leftarrow \text{cert}(m_{u-1}^i) \in m_u^i$
        \EndFor
    \EndIf
\EndProcedure
\end{algorithmic}
\end{algorithm}
\begin{algorithm}[t]
\footnotesize
\caption{APM-BRB (Deliver Condition 2)}
\label{algo:deliver_part2}
\begin{algorithmic}[1]
\State \textbf{For each node in $\Pi$}
\State $\mathcal{M} \leftarrow \emptyset$
\Upon{obtaining $\text{cert}(m_r^j)$ for each $P_j \in \Pi$ in round $r$}
    \State $N_c \leftarrow \textsc{sampleCommittee}(j, r)$
    \State Collect $m_r^j$ from $N_c$ s.t. $\text{cert}(m_r^j)$ is the $(\frac{n_c}{2}+1, N_c)$ threshold signature for $m_r^j$
    \State $\mathcal{M} \leftarrow \mathcal{M} \cup \{m_r^j\}$
    \If{$|\mathcal{M}| = n$}
        \For{each $\text{cert}(m^i_{r-1})$ that exists in $\text{prev}\{\}$ of all $m \in \mathcal{M}$}
            \For{$u = r-1$ \textbf{down to} $1$}
                \State \textsc{Deliver}$_i(m_u^{i}, u)$ if not already delivered \label{algo:deliver_part2: con 1}
            \EndFor
        \EndFor
    \EndIf
\EndUpon
\State
\State \textbf{For each node in $\Pi$}
\Upon{\textsc{Deliver}$_j(m^j_r, r)$ for $2f+1$ unique nodes $P_j$}
    \For{each $\text{cert}(m_{r-1}^i)$ that appears in $\text{prev}\{\} \in m_r^j$ of all $2f+1$ messages $m_r^j$} \label{line: totality_2}
        \State $tempCert \leftarrow \text{cert}(m_{r-1}^i)$
        \For{$u = r-1$ \textbf{down to} $1$}
            \State $N_c \leftarrow \textsc{sampleCommittee}(i, u)$
            \State obtain $m_u^i$ from $N_c$ s.t. $tempCert$ is the $(\frac{n_c}{2}+1, N_c)$ threshold signature for $m_u^i$
            \State \textsc{Deliver}$_i(m_u^{i}, u)$ if not already delivered \label{algo:deliver_part2: con 2}
            \State $tempCert \leftarrow \text{cert}(m_{u-1}^i) \in m_u^i$
        \EndFor
    \EndFor
\EndUpon
\end{algorithmic}
\end{algorithm}

We dub our algorithm Amortized Probabilistic Multi-shot BRB (APM-BRB), and describe it in parts: \cref{algo:our_algo,algo:our_algo_send,algo:deliver_part1,algo:deliver_part2}.

The algorithm progresses in rounds, where each node $P_i$ attempts to reliably broadcast their message $m_r^i$ to $\Pi$. 

\subsubsection{Sending Operation}
Besides round $r=1$, in every other round $(r>1)$, each sender $P_i$ sends a message $m^i_r$ only when it has obtained a $(\frac{n_c}{2}+1, N_c)$ threshold signature (or cert($m^i_{r-1}$)) for the message it has sent in the previous round (\cref{algo:our_algo_send: message params} of \cref{algo:our_algo_send}); this is generated by obtaining $\frac{n_c}{2}+1$ unique signature shares from nodes from $\textsc{sampleCommittee}(i,r-1)$ (\cref{algo:our_algo_send:signature share} of \cref{algo:our_algo_send}).
As a result, each honest node $P_i$ creates a causal chain of messages, where each message $m_r^i$ references the message it has sent in round $r-1$ respectively.

Each message also contains some additional metadata (prev$\{\}$,triggers$[]$) of size $O(nk)$.
In $m^i_r$, prev$\{\}$ contains certificates from $2f+1$ distinct nodes in the  immediate previous round that has a causal chain consistent with its own local memory (promises[][]). Essentially, prev$\{\}$ is a signal to other nodes to show which causal chains (if equivocation exists) with certificates in round $r-1$, that a particular honest node $P_i$ will extend exclusively. 
On the other hand, triggers$[]$ serves as a vehicle for honest nodes to indicate to other nodes messages they have called $\textsc{Deliver}(\cdot,\cdot)$ on; this is crucial for totality. 

When an honest node $P_j$ receives a message $(m^i_r,r,i)$, it checks whether it should participate in this committee by invoking the public $\textsc{sampleCommittee}()$ function (\cref{algo:our_algo}). 
If $P_j$ has not already received another message from $P_i$ for the same round, it will perform another check utilizing the $\textsc{syncMsg}()$ function (\cref{algo:our_algo}). 
The $\textsc{syncMsg}()$ function updates the promises[] memory object, and returns a message to check if the received message indeed extends the causal chain as seen in promises$[i][r-1]$ (\cref{algo:our_algo_send: check extend} of \cref{algo:our_algo_send}). 

Importantly, once an honest node writes a particular message into its local memory object promises$[i][r]$, it becomes immutable. This immutability ensures that an honest node will not extend causal chains that it has deemed to equivocate within its local memory. 
In other words, Byzantine nodes might create distinct chains (since it is possible to obtain $\frac{n_c}{2}+1$ unique certificates every round); however, each honest node will only aid in extending one of them. 

Before sending a signed acknowledgment for the message (\cref{algo:our_algo_send:signature share} of \cref{algo:our_algo_send}), the honest node also re-sends the message to the committee. This action prevents a byzantine sender from selectively choosing which honest nodes in the committee may participate in the creation of the cert(). 
If the byzantine node performs such an action, it is akin to a network adversary controlling the order in which messages may be delivered. 
We show in \cref{thm: Impossibility} that if this is possible, agreement cannot be achieved. 

\subsubsection{Delivery Mechanism}

APM-BRB has two possible $\textsc{Deliver}(\cdot,\cdot)$ conditions. The first condition describes the common case (\cref{algo:deliver_part1}):
Once $P_j$ receives a certificate (cert$(m^i_{r})$) (via \cref{algo:deliver_part1:process 1,algo:deliver_part1:process 2} of \cref{algo:deliver_part1}) for a message sent by $P_i$ in round $r:r\geq \varphi$. It calls $\textsc{processCertificate}()$, which works backwards to derive the message for $m_{r-\varphi}^i$. Regardless of whether this chain equivocates with $P_j$'s promises$[i]$. Since such a certificate exists, it indicates that with probability at least $1-\varepsilon$, a set of $f+1$ unique honest nodes have endorsed $m_{r-\varphi}^i$ by having it in their respective promises$[i][r-\varphi]$.

This is sufficient to deliver $m_{r-\varphi}^i$ (more precisely to call $\textsc{Deliver}_i(m_{r-\varphi}^i,r-\varphi)$), as by the quorum intersection property, another set of $f+1$ honest nodes that endorses another message from $P_i$ for round $r-\varphi$  cannot exist.

The second condition describes an optimistic case (\cref{algo:deliver_part2}): When an honest node obtains a message from each of the $n$ nodes in a single round (for a total of $n$ messages), and discovers that a particular message $m_{r-1}^j$ is referenced in every set of prev$\{\}$. The node can be sure that no other message from $P_j$ in round $r-1$ may receive the sufficient $\varphi$ chain of certificates to be delivered, ensuring that if any node calls $\textsc{Deliver}_j(m',r-1)$, then $m'=m_{r-1}^j$. Upon witnessing these conditions, a node can simply call $\textsc{Deliver}_j(m_{r-1}^j,r-1)$ (\cref{algo:deliver_part2: con 1} in \cref{algo:deliver_part2}).

Since every honest node's message will be delivered by all honest nodes eventually, all nodes will eventually witness $m_{r-1}^j$ in the prev$\{\}$ of at least $2f+1$ delivered messages from round $r$. Therefore, if an honest node has delivered a message via \cref{algo:deliver_part2: con 1} of \cref{algo:deliver_part2}, then all honest nodes will eventually be able to deliver the same message with \cref{algo:deliver_part2: con 2} in \cref{algo:deliver_part2}. 
\section{Correctness of APM-BRB}
\label{sec: proofs}

In this section, we will formally prove that APM-BRB satisfies \cref{def: PMBRB}, assuming that $\varphi$ and $n_c$ are chosen sufficiently. We will demonstrate discrete pessimistic bounds for the values of $\varphi, n_c$ in \cref{thm:success-prob}. Therefore for simplicity, We will assume throughout the rest of this section that for every committee sampled has a majority of honest nodes. We also assume that for any round $r$, for all rounds in between $r$ and $r + \varphi$, at least $f+1$ honest nodes sign at least one message in these rounds. We will prove later that these assumptions will hold with probability at least $1 - \varepsilon$.



\begin{lemma}
\label{lemma: agreement}
    Deliver Condition 1 of APM-BRB satisfies Agreement with probability at least $1-\varepsilon$.
\end{lemma}
\begin{proof}
    Suppose that for some round $r$, an honest node calls $\textsc{Deliver}_i(m_r^i,r)$ when cert($m_{r+\varphi}^i$) is received. From \cref{line:deliver 1_1} to \cref{line:deliver 1_2} of \cref{algo:deliver_part1}, there must exist a chain of messages spanning from $m_r^i$ to $m^i_{r+\varphi}$ where each message $m^i_v: v \in (r,r+\varphi]$ contains a valid certificate/threshold signature for $m^i_{v-1}$. In each round, when an honest node in the committee signs the message, by the definition of \textsc{syncMsg} (\cref{algo:our_algo}), either of the following holds: it has $m_r^i$ in its variable promises$[i][r]$, or it adds $m_r^i$ to promises$[i][r]$ (\cref{line: promises_1} of \cref{algo:our_algo}). To see why the first holds, suppose it has another message $(m')_r^i$, then in \textsc{syncMsg} it must have a message $(m')_{r+1}^i$ in promises$[i][r+1]$, as $m_{r+1}^i$ contains a certificate for $m_r^i$ and not $(m')_r^i$. This argument can be repeated until round $v$.

    As we have that at least $f+1$ unique honest nodes must participate in signing for at least one of these $\varphi$ many rounds, that means that at least $f+1$ unique honest nodes have $m_r^i$ in their local memory, i.e. promises$[i][r]$.

    Suppose then that a different honest node calls $\textsc{Deliver}_i((m')_r^i,r)$ when cert($(m')_{r+\varphi}^i$) is received. Then by similar argument, there must exist chain of $\varphi$ messages that spans from  $(m')_r^i$ to  $(m')_{r+\varphi}^i$ where each $(m')_v^i$ contains a threshold signature for $(m')_v^i$, and therefore there must exist at least $f+1$ honest nodes that have $(m')_r^i$ in their local memory.

    By the quorum intersection property, it is impossible that $f+1$ unique honest nodes have $m_r^i$ in their promises$[i][r]$, and another disjoint set of $f+1$ unique honest nodes that have $(m')_r^i$ in their promises$[i][r]$.


    Therefore, for any round $r$, for two different honest nodes that calls $\textsc{Deliver}_i(m_r^i,r)$ and $\textsc{Deliver}_i((m')_r^i,r)$, it holds that $(m')_r^i = m_r^i$, satisfying agreement.


\end{proof}

\begin{lemma}
    Deliver Condition 2 of APM-BRB satisfies Agreement with probability at least $1-\varepsilon$.
\end{lemma}

\begin{proof}

    Consider when an honest node has delivered a message due to \cref{algo:deliver_part2: con 1} of \cref{algo:deliver_part2}. It must have received $n$ messages, each with cert($m_r^j$) in the prev$\{\}$ set. This implies that all honest nodes must have $m_r^j$ in their respective promises$[j][r]$. Consider any different message, $(m')_r^j$, for node $j$ in round $r$. By similar argument as~\cref{lemma: agreement}, any such message $(m')_r^j$ may never meet the threshold to obtain a valid certificate. Therefore, $(m')_r^j$ will not be delivered by any honest node.
    Additionally, observe that if all honest nodes have identical messages for a node $j$ contained in their local memory for some round $r$ (that is, in promises$[j][r]$), it must hold that all promises$[j][r]$ to promises$[j][1]$ are identical across all honest nodes as well. Therefore, since no equivocating messages may exist, delivering promises$[j][r]$ to promises$[j][1]$ is safe and no honest node can deliver a different message.

\end{proof}

\begin{lemma}
\label{lemma: liveness}
    APM-BRB satisfies Validity with probability at least $1-\varepsilon$.
\end{lemma}
\begin{proof}
    Firstly, each node will be able to send a message every round. Since all honest nodes must have sent a message for round 1, each node must be able to complete $\textsc{reqMsg}()$ from by querying $2f+1$ honest nodes eventually. This holds for all rounds.

    Consider any honest node $P_i$. Recall that we have assumed that each committee of size $n_c$ has at least ($\frac{n_c}{2}+1$) of honest nodes. In round $r$, this party of honest nodes will sign $m_r^i$ and thus forms a sufficient threshold to generate cert($m_r^i$). Now, for round 2, $P_i$ has cert$(m_1^i)$ and can thus include it in $m_2^r$ as in~\cref{algo:our_algo_send} and send it to the committee, which also has sufficient honest nodes which sign the message and thus cert$(m_2^i)$ can be formed. This repeats for all rounds. Therefore for any round $r$, eventually every honest node can form cert($m_{r'}$), for all $r \le r' \le r+\varphi$ and so all honest nodes will be able to call $\textsc{deliver}_i(m,r)$ via~\cref{algo:deliver_part1}.

\end{proof}

\begin{lemma}
\label{lemma: liveness}
    APM-BRB satisfies Totality with probability at least $1-\varepsilon$.
\end{lemma}

\begin{proof}

    Consider the case that the sender $i$ is honest.
    Now, suppose some honest node calls $\textsc{deliver}_i(m,r)$. This implies that this honest node has the message $m$ and the certificate cert$(m_{r+\varphi}^i)$. Since messages are eventually delivered, all honest nodes will eventually also obtain all the certificates for messages between rounds $r$ and $r + \varphi$. By an argument identical to~\cref{lemma: liveness}, all honest nodes will eventually call $\textsc{deliver}_i(m,r)$, thereby satisfying totality.

    Consider the case where the sender $P_i$ is Byzantine. Suppose some honest node calls $\textsc{deliver}_i(m,r)$, it will also perform \cref{line: trigger_1} of \cref{algo:deliver_part1} where it updates its own triggers[] with cert($m_{r+\varphi}^i$); that caused it to perform $\textsc{deliver}_i(m, r)$. By validity, the next message it sends (which will include cert($m_{r+\varphi}^i$) as part of triggers[]) will eventually be delivered by all honest nodes. Therefore, all honest nodes will eventually perform \cref{line:trigger_2} of \cref{algo:deliver_part1}. Further, since at least one honest node must have received and signed the message sent by $P_i$ in each round between $r$ and $r + \varphi$, honest nodes can successfully retrieve the messages and eventually perform $\textsc{deliver}_i(m,r)$. Therefore totality is satisfied.

\end{proof}



\begin{lemma}
    APM-BRB satisfies Completion Sequentiality with probability at least $1-\varepsilon$.
\end{lemma}
\begin{proof}
When $P_j$ calls $\textsc{deliver}_i(m, r)$ for a message sent by $P_i$ where $i \neq j$, it recursively locates all messages previously sent by $P_i$ via the threshold signature embedded in $m$ (\cref{line: totality_1} of \cref{algo:deliver_part1} and \cref{line: totality_2} of \cref{algo:deliver_part2}). Observe that for the certificate of the message for any round to exist, it must have been signed by at least one honest node in $N_c$ (for that round).

Since $\textsc{sampleCommittee}()$ is a public function, $P_j$ can identify and query the nodes in $N_c$, and is guaranteed to be able to eventually obtain the message $m'$ that $P_i$ sent in round $r-1$ from an honest node that participated in signing it. The authenticity of $m'$ can then be verified against the threshold signature contained in $m$.

Consequently, $P_j$ is able to invoke $\textsc{deliver}_i(m', r-1)$, thereby satisfying completion sequentiality.

\end{proof}

\section{Complexity Analysis of APM-BRB}
\label{sec: complexity analysis}

In this section, we show that our choices of $n_c$ and $\varphi$ are sufficient to give the probability of failure of our algorithm to be at most $\varepsilon$. Namely, we prove the following.

\begin{theorem}
\label{thm:success-prob}
    Let $\varphi = 2 n \log \frac{1}{\varepsilon}$ and $n_c = 18\log \frac{\phi}{\varepsilon}$. Then with probability at least $1 - \varepsilon$, the following holds for every sender $P_i$. First, every sampled committee (of the $\phi$ many for each sender) has a majority of honest nodes. Second, among these $\phi$ sampled committees, at least $f+1$ unique honest nodes must sign at least one message sent by $P_i$ among all of the rounds.
\end{theorem}

Note that all logarithms in this section are assumed to be the natural logarithm. 
We first show that our choice of $n_c$ is sufficient, in particular, we ensure that the probability that any committee does not have a majority of honest nodes is at most $\frac{\varepsilon}{2 \phi}$.





\begin{lemma}
\label{lemma:majority-honest-committee}
    Consider an $\varepsilon > 0$. Let $N$ be a set of nodes sampled uniformly at random without replacement of size $n_c$, where $n_c \geq 18 \log (\frac{2 \phi}{\varepsilon})$. The probability that there are more than $\frac{n_c}{2}$ byzantine nodes sampled in $N$ is at most $\frac{\varepsilon}{2\phi}$.
\end{lemma}
\begin{proof}
Consider a uniform sample of $n_c$ nodes without replacement. Let $X_i$ be an indicator random variable that takes on $1$ if the $i$th sampled node is honest, and 0 otherwise, and let $X = \sum_{i = 1}^{n_c} X_i$.

As the number of honest nodes is $2f+1$ out of $3f$ total nodes, it holds that:

    
\[ \Pr[X_i = 1] \ge 2/3 .\]

Therefore by linearity of expectation,

\begin{equation*}
    \mathbb{E}(X) \ge \frac{2 n_c}{3}.
\end{equation*}

By applying the standard Hoeffding's bound\cite{Hoeffding1963ProbabilityIF}, which holds for sampling without replacement, we obtain:

\begin{align*}
    \Pr \left[ |X - \mathbb{E}(X)| \ge \frac{n_c}{6} \right] & \le \exp \left(-\frac{2(n_c/6)^2}{n_c} \right) \\
    \le  \exp \left(-\frac{n_c}{18} \right) .
\end{align*}

Let $B$ be the event that at $X$ is less than $\frac{n_c}{2}$, that is, at least $\frac{n_c}{2}$ byzantine nodes are sampled in the set $N$. Then it is clear that $\Pr(B) \le \exp (-\frac{n_c}{18})$. As we have that $n_c \ge 18 \log \frac{2\phi}{\varepsilon}$, we have that:

\begin{align*}
    \Pr(B) & \le \exp (-\frac{n_c}{18}) \\
    & \le \exp (- \log \frac{2 \phi}{\varepsilon}) \\
    & \le \frac{\varepsilon}{2 \phi} .
\end{align*}

Therefore, with probability at most $\frac{\varepsilon}{2\phi}$, we have that the sampled set consists of less than $\frac{n_c}{2}$ honest nodes.
\end{proof}

Now, suppose that every round indeed has a simple majority of honest nodes. We show that if we take a sufficient number of committees, we will have at least $f+1$ unique honest nodes whose signatures of the message forms the threshold certificate in at least one of the rounds.

Before we prove this, we will need a concentration bound on the sum of independent geometric random variables.
\begin{theorem}[\cite{janson2018tail}]
    \label{thm:tail-exponential}
    Consider random variables $X_1, ..., X_n$ where $X_i$ is drawn from a geometric distribution with probability parameter $p_i$ and let $X = \sum_{i = 1}^n X_i$. Then for any $\lambda > 1$, it holds that:
    \[ \Pr(X \ge \lambda \mu) \le \exp(-p^*\mu(\lambda - 1 - \log \lambda) )\] 
    
    where $\mu := \mathbb{E}(X)$ and $p^* = \min_{1 \le i \le n} p_i$. 
\end{theorem}

Now, we show that if $\varphi = O(n)$, then the number of committees is indeed sufficient for our protocol to succeed with high probability.
\begin{lemma}
\label{lemma:coupon-collector}
    Consider an $\varepsilon > 0$ and let $\varphi = 2 n \log \frac{1}{\varepsilon}$. Fix an honest node $i$ and consider $\varphi$ many rounds of the protocol, where each committee has a majority of honest nodes. Then with probability at least $1 - \varepsilon^2$, there are at least $f+1$ unique honest nodes which sign at least one message sent from $i$ throughout these $\varphi$ rounds.
\end{lemma}

\begin{proof}
    Fix some honest node $i$. For $1 \le a \le f+1$, let $Y_a$ be a random variable that denotes the number of rounds needed before the $a$th new honest node signs a message sent by $i$ within these $\varphi$ rounds. Observe that $Y_a$ is a geometric random variable where $\Pr(Y_a = k) = p_a(1-p_a)^{k-1}$ for $p_a = \frac{2f+2-a}{2f+1}$. Let $Y = \sum_{a \le f+1} Y_a$. We have that $\frac{n}{2} \le \mathbb{E}(Y) \le 2n$, where the lower bound follows from the fact that clearly $\mathbb{E}(Y_a) \ge 1$ for all $a$. The upper bound comes from the fact that $\min p_a = \frac{f+1}{2f+1} \ge 0.5$, and so $\mathbb{E}(Y_i) \le 2$ for all $a$.

    By Theorem~\ref{thm:tail-exponential}, we have that:
    \[\Pr(Y \ge \lambda \mu) \le \exp(-0.5 \mu(\lambda-1-\log \lambda)) \] 
    where $\mu := \mathbb{E}(Y)$.

    Now, let $\lambda = \log \frac{1}{\varepsilon}$, and along with the fact that $\mu \le 2n$ gives:

    \begin{align*}
        \Pr(Y \ge 2 n \log \frac{1}{\varepsilon}) & \le \exp \left(-0.5 \mu (\log \frac{1}{\varepsilon} - 1 - \log \log \frac{1}{\varepsilon}) \right) \\
        & \le \exp \left(-0.5 \mu (0.5 \log \frac{1}{\varepsilon}) \right) && \text{(For sufficiently small $\epsilon < 0.001$)} \\
        & \le \exp \left(-1/8 \cdot n \log \frac{1}{\varepsilon} \right) && \text{(As $\mu > 0.5n$)} \\
        & \le \varepsilon^{n/8} \\
        & \le \varepsilon^{2} && \text{(For $n \ge 16$).}\\
    \end{align*}

    Since $Y$ is a random variable that denotes the number of rounds before $f+1$ unique honest nodes sign at least one of the messages in the $\varphi$ consecutive rounds, this completes the proof of the lemma.
\end{proof}

Now, we can complete the proof of~\cref{thm:success-prob}.

\begin{proof}[Proof of~\cref{thm:success-prob}]
    Fix any honest sender node $i$. 
    
    Observe that by~\cref{lemma:majority-honest-committee} and the selection of $n_c$, each committee has probability at most $\frac{\epsilon}{2\phi}$ that it does not have a simple majority of honest nodes. Now, our protocol will sample $\phi$ distinct committees. By union bound, every committee in the execution of our protocol has a majority of honest nodes with probability at least $1 - \frac{\varepsilon}{2}$.

    Observe that for any round $r$, the set of committees from round $r$ to $r + \varphi$ is identical, as our protocol uses $r \bmod \varphi$. By~\cref{lemma:coupon-collector} and the choice of $\varphi$, given that each has a majority of honest nodes, then it must be that at least $f+1$ unique honest nodes signed at least one of the messages from the sender $i$ with probability at least $1 - \varepsilon^2$. 
    
    Taking a union bound over the failure probability of the two above events, we have that for a fixed sender $i$, our protocol will succeed with probability at least $1 - \varepsilon$.
\end{proof}

\subsection{Communication Cost Breakdown for APM-BRB}
\label{Comm cost breakdown}

Let us now compute the communication cost of APM-BRB.
A sender performs $n$ instances of $\textsc{reqMsg}()$ before constructing its message to send: where the sender will query $n$ nodes to obtain at least $2f+1$ messages from the previous round. This incurs a cost of $n+n(|m|+nk+k)$.

The sender incurs a cost of $n_c (|m|+nk)$ when sending its message along with $nk$ bits of metadata (triggers$[]$, \text{prev}$\{\}$) to a committee (lines 1 to 17 of \cref{algo:our_algo_send}). 

Each node in $N_c$ then runs $\textsc{syncMsg}()$ incurring a cost of $n_c^2 + n_c^2(|m|+nk)$. Each node then forwards the message again to all other nodes in $N_c$ before broadcasting its signature share to all nodes, incurring a cost of $n_c^2(|m|+nk)+nn_c k$.



The total cost of each round is therefore $n+n(|m|+nk+k) +n_c (|m|+nk) + n_c^2 + n_c^2(|m|+nk) + n_c^2(|m|+nk)+nn_c k$, which can be bounded by $O(n|m|+n^2k+n_c^2|m|+nn_c^2k)$.

Recall that from~\cref{lemma:majority-honest-committee}, we have set $n_c = O(\log n)$, therefore the cost of each round can be bounded by $O(n|m|+n^2k)$
For brevity, let $\Delta$ denote the cost of one round.


For a particular message to complete reliable broadcast, a chain of $\varphi$ certificates must be constructed. Therefore, the total cost to reliably broadcast the first message using APM-BRB is $\varphi \cdot \Delta$. Each subsequent reliable broadcast incurs only $\Delta$ per round.

By \cref{def: amortized}, the cost to broadcast $r$ messages is:
\[
 C(r,n,f) = \varphi \cdot \Delta + r \cdot \Delta
\]

Our amortized cost can be bounded by:
\[
\frac{C(r,n,f)}{r} = \frac{\varphi}{r} \cdot \Delta +\Delta
\]

For a fixed $n$, as $r\rightarrow \infty$ and $\varphi = O(n)$, it holds that $\frac{\varphi}{r} = o(1)$, and so $\frac{C(r,n,f)}{r} = O(\Delta) = O(n|m|+n^2k)$.

And so we obtain:
\[
\frac{C(r,n,f)}{r} = O(n|m|+n^2k)
\]

When the message size dominates the metadata, that is $|m| = \Omega(nk)$, the amortized communication complexity simplifies to $O(n|m|)$.

We also make the following remark on the latency. In the common case per \cref{algo:deliver_part1}, delivery takes $\varphi=O(n)$ rounds. Our optimistic delivery mechanism described in \cref{algo:deliver_part2} achieves $O(1)$ round latency in favorable conditions.




\section{Impossibility of Optimally Resilient Committee-based BRB in Asynchrony}
\label{sec: Impossibility}



Nayak et al.~\cite{Nayak2020ImprovedEP} showed that any deterministic BRB algorithm that tolerates $\Theta(n)$ byzantine nodes has a communication lower bound of $\Omega(n|m|+n^2)$ even in synchrony.  
Currently there is no known lower bound on the communication for optimally resilient \emph{probabilistic} BRB algorithms. However, we believe that applying sampling and amortization is unable to break below this bound in a fully asynchronous network, in particular, any such protocol would need to operate fundamentally differently from ours.
\begin{theorem}
\label{thm: subcommittee-asynchrony}
    An optimally resilient, committee based Probabilistic BRB is not possible in full network asynchrony.
\end{theorem}

\begin{definition}[Committee-based BRB]
    A BRB algorithm is regarded as Committee based, if it restricts the acquiring of $f+1$ authenticated acknowledgments from distinct honest nodes to communication with a collection of sub-linear sized node committees, where each committee contains a simple majority of honest nodes with probability at least $1-\varepsilon$.
\end{definition}

\begin{lemma}
\label{lemma: authack}
    During an execution of a BRB algorithm, a node must obtain an authenticated acknowledgment from at least $f+1$ distinct honest nodes to safely \textsc{Deliver(m)}.
\end{lemma}

\begin{proof}
    Suppose that in the protocol, \textsc{Deliver(m)} is performed for some message $m$ when authenticated acknowledgment from at most $f$ distinct honest nodes is obtained. Then observe that there can exist another message $m'$, where \textsc{Deliver(m')} is performed when authenticated acknowledgment from a disjoint set of $f$ distinct honest nodes. This causes violation of agreement.
\end{proof}

\begin{theorem}
\label{thm: Impossibility}
    Consider a node $P_i$ sending message $m$ to uniformly sampled committees of size $n_c = o(n)$, and with probability at least $1-\varepsilon$, the committees contain a simple majority of honest nodes. 
    From each committee, $P_i$ may only receive at most $\frac{n}{2}+1$ acknowledgments.
    Then is not possible for the node to receive $f+1$ unique honest acknowledgments for the message with probability $1-\varepsilon$ in network asynchrony.  
\end{theorem}
We only consider committees of size $o(n)$, as otherwise, there is no asymptotic improvement in communication complexity to send a message to the committee as compared to all nodes. 
\begin{proof}
    As each committee is of size $o(n)$ and $f+1$ honest nodes (i.e., $O(n)$ many nodes) are needed to achieve agreement via the quorum intersection argument, this implies that any such protocol needs to go through multiple rounds of sending messages to achieve this.

    Now, pick any arbitrary subset $B$ of $f$ honest nodes. Observe that in any committee, the expected proportion of nodes which are byzantine or among this set $B$ is $\frac{2f}{3f+1}$ (and is strictly above half of the committee). Further, the expected number of honest nodes in the committee is also $\frac{2f+1}{3f+1}$. Therefore, if there is a majority of honest nodes with probability at least $1 - \varepsilon$, there is also a majority of nodes which are byzantine or from $B$.
    
    In a fully asynchronous network, the adversary may arbitrarily delay messages to and from nodes in the committee such that messages to byzantine nodes or nodes in $B$ are sent/received, while messages to other nodes are arbitrarily delayed. This causes the set of honest nodes which is known to have acknowledged the message to be bounded at $f$. 
\end{proof}

The proof of~\cref{thm: subcommittee-asynchrony} then follows directly from~\cref{thm: Impossibility} and~\cref{lemma: authack}.


\section{Discussion on Byzantine Atomic Broadcast}

\label{sec: discussion}

A practical use case for BRB protocols is serving as a communication primitive in a Byzantine Atomic Broadcast (BAB) ~\cite{Keidar2021AllYN,Danezis2021NarwhalAT,Giridharan2022BullsharkDB} algorithm. Keidar et al.~\cite{Keidar2021AllYN} popularized a Direct Acyclic Graph (DAG) construction, utilizing it to instantiate a common-core~\cite{Canetti1993FastAB}. This construction outperforms conventional agreement-on-common set~\cite{ben-or,bkr,Cachin2001SecureAE} BAB algorithms~\cite{honeybadger,Lu2020DumboMVBAOM,Duan2018BEATAB,Liu2020EPICEA} in expected latency and communication. 

In these DAG-based BAB protocols, nodes reliably broadcast a message every logical round. Upon BRB completion, the message is considered a \textit{block.} Each block references a threshold of at least $2f+1$ blocks (from other nodes) from the previous logical round, forming a DAG. 
BAB is achieved when certain blocks are selected to be leaders by means of a global perfect coin abstraction~\cite{Boneh2001ShortSF,Libert2014BornAR,Shoup2000PracticalTS}, typically instantiated by means of a threshold signature; where leaders are then totally ordered along with blocks in its causal history (blocks which it references). Although not stated as a necessary property~\cite{Keidar2021AllYN,Danezis2021NarwhalAT,Giridharan2022BullsharkDB}, \textit{strict sequentiality} of the multi-shot BRB is implied. That is, the BRB instance in round $r$ must terminate before the one in $r+1$ may begin. 

Therefore, our question is if its possible to instantiate a similar DAG-BAB algorithm utilizing APM-BRB; A BRB that has completion sequentiality instead. A simple construction is to have the BRB instance from round $r$ reference BRB instances from round $r-\varphi$, effectively creating $\varphi$ parallel DAG instances. A recent work~\cite{Arun2024ShoalHT} utilizes a parallel DAGs and may capitalize such a construction. 

However, it is still an open question if a BAB algorithm could be instantiated while utilizing APM-BRB to have BRB instances from round $r$ reference those from round $r-1$, as out of $2f+1$ references, $f$ of them might not ever complete. These ambiguous references make it difficult to select leaders out of the blocks, and determine what blocks are definitively part its causal history. Nonetheless, we believe it is still interesting to explore this direction; as it opens the door to extremely bandwidth efficient practical deployments of BAB. 

\section{Conclusion}
\label{sec: conclusion}

In this paper, we present APM-BRB, a Byzantine Reliable Broadcast protocol that achieves optimal resilience under a random asynchronous network. By employing rotating committees and amortization across multiple broadcasts, APM-BRB achieves an amortized communication complexity of $O(n|m|+n^2k)$. When the message size dominates the metadata, i.e., $|m| = \Omega(nk)$, this simplifies to an optimal $O(n|m|)$. APM-BRB is the first BRB protocol to achieve this with optimal resilience under asynchrony.
In the common case, APM-BRB terminates in $O(n)$ rounds, but our optimistic delivery mechanism enables termination in $O(1)$ rounds under favorable conditions.
Furthermore, we show an impossibility result: An optimally resilient, committee-based BRB is not possible in full asynchrony.

\bibliography{biblo}

\appendix

\end{document}